\documentclass[doublespace]{article}
\usepackage{arxiv}

\usepackage[utf8]{inputenc} 
\usepackage[T1]{fontenc}    
\usepackage{hyperref}       
\usepackage{url}            
\usepackage{booktabs}       
\usepackage{amsfonts}       
\usepackage{nicefrac}       
\usepackage{microtype}      
\usepackage{lipsum}		
\usepackage{graphicx}

\usepackage{moreverb}
\usepackage{pdflscape}
\usepackage{graphicx}
\usepackage{enumitem}

\usepackage{array}
\usepackage{mathtools} 
\usepackage{chngpage}

\DeclareMathAlphabet{\mathbfsf}{\encodingdefault}{\sfdefault}{bx}{sl}

\usepackage{tikz}
\usetikzlibrary{mindmap}			
\usetikzlibrary{arrows,arrows.meta, calc,positioning, fit}
\usetikzlibrary{shapes,decorations, decorations.pathreplacing}
\definecolor{mylightskyblue}{rgb}{0.53,0.81,0.98}
\definecolor{mydodgerblue3}{rgb}{.09,0.45,0.80}
\definecolor{mynavy}{rgb}{0, 0, 0.5}

\usepackage{tabularx}

\usepackage{amsmath, stackrel, bm}
\usepackage[font=normalsize]{subcaption}
\usepackage{multirow}
\usepackage{acro}
\acsetup{first-style=long-short}
\DeclareAcronym{ATE}{
  short = ATE,
  long  = average treatment effect
}
\DeclareAcronym{ATT}{
  short = ATT,
  long  = average treatment effect of the treated 
}
\DeclareAcronym{CART}{
  short = CART,
  long  = Classification And Regression Tree
}
\DeclareAcronym{DAG}{
  short = DAG,
  long  = directed acyclic graph
}
\DeclareAcronym{IPTW}{
  short = IPTW,
  long  = Inverse Probability of Treatment Weighting
}
\DeclareAcronym{GBSG}{
  short = GBSG,
  long  = German Breast Cancer Study Group
}
\DeclareAcronym{MFP}{
  short = MFP,
  long  = Multivariable Fractional Polynomial
}\DeclareAcronym{MOB}{
  short = MOB,
  long  = Model-based Recursive Partitioning for Subgroup Analysis
}
\DeclareAcronym{MSE}{
  short = MSE,
  long  = Mean Squared Error
}
\DeclareAcronym{palmtree}{
  short = palmtree,
  long  = partially additive (generalized) linear model tree
}
\DeclareAcronym{RFS}{
  short = RFS,
  long  = Relapse-free survival
}\DeclareAcronym{PI}{
  short = PI,
  long  = permutation importance
}

\usepackage[numbers, super]{natbib}
\newcommand\BibTeX{{\rmfamily B\kern-.05em \textsc{i\kern-.025em b}\kern-.08em
T\kern-.1667em\lower.7ex\hbox{E}\kern-.125emX}}

\title{Decomposition of Explained Variation in the Linear Mixed Model}

\author{{Nicholas Schreck}\\ 
	Biostatistics Division\\
	German Cancer Research Center\\
    D-69120 Heidelberg, GERMANY \\
	\texttt{nicholas.schreck@dkfz-heidelberg.de}
	\And
	{Manuel Wiesenfarth} \\
	Biostatistics Division\\
	German Cancer Research Center\\
    D-69120 Heidelberg, GERMANY \\
}
\newcommand{\headeright}{Schreck and Wiesenfarth}

\renewcommand{\shorttitle}{Explained Variation in the LMM}

\newtheorem{theorem}{\sc Theorem}[section]
\newtheorem{lemma}[theorem]{\sc Lemma}

\begin{document}

\maketitle

\newcommand{\RSS}{\mathop{\textsc{RSS}}}
\newcommand{\TSS}{\mathop{\textsc{TSS}}}
\newcommand{\ESS}{\mathop{\textsc{ESS}}}
\newcommand{\E}{\mathbb{E}}
\newcommand{\No}{\mathcal{N}}
\newcommand{\R}{\mathbb{R}}
\newcommand{\Var}{\mathrm{var}}
\newcommand{\Cov}{\mathrm{cov}}
\newcommand{\tr}{\mathop{\mathrm{tr}}}
\newcommand{\diag}{\mathop{\mathrm{diag}}}
\newcommand{\ones}{\mathbf{1}_n}
\newcommand{\onest}{\mathbf{1}_n^{\top}}
\newcommand{\Id}{\mathbf{I}_n}
\newcommand{\Idp}{\mathbf{I}_p}
\newcommand{\se}{\sigma_{\boldsymbol{\varepsilon}}^2}
\newcommand{\seh}{\dot{\sigma}_{\boldsymbol{\varepsilon}}^2}
\newcommand{\su}{\sigma_{\boldsymbol{u}}^2}
\newcommand{\sui}{\sigma_{\boldsymbol{u}_i}^2}
\newcommand{\suone}{\sigma_{\boldsymbol{u}_1}^2}
\newcommand{\sur}{\sigma_{\boldsymbol{u}_m}^2}
\newcommand{\suj}{\sigma_{\boldsymbol{u}_j}^2}
\newcommand{\gi}{\gamma_{i}^2}
\newcommand{\suh}{\dot{\sigma}_{\boldsymbol{u}}^2}
\newcommand{\suih}{\dot{\sigma}_{\boldsymbol{u}_i}^2}
\newcommand{\syh}{\hat{\sigma}_{\boldsymbol{y}}^2}
\newcommand{\Xf}{\mathbf{X}}
\newcommand{\tXf}{\tilde{\mathbf{X}}}
\newcommand{\tXft}{\tilde{\mathbf{X}}^{\top}}
\newcommand{\Af}{\mathbf{A}}
\newcommand{\Afd}{\dot{\mathbf{A}}}
\newcommand{\Xft}{\mathbf{X}^{\top}}
\newcommand{\Zft}{\mathbf{Z}^{\top}}
\newcommand{\Zf}{\mathbf{Z}}
\newcommand{\Gf}{\mathbf{G}}
\newcommand{\Gfd}{\dot{\mathbf{G}}}
\newcommand{\Cf}{\mathbf{C}}
\newcommand{\Df}{\mathbf{D}}
\newcommand{\Dfd}{\dot{\mathbf{D}}}
\newcommand{\Pfh}{\mathbf{P}_{H}}
\newcommand{\Pfhd}{\dot{\mathbf{P}}_{H}}
\newcommand{\Pfhc}{\mathbf{P}_{H}^{C}}
\newcommand{\Pfhcd}{\dot{\mathbf{P}}_{H}^{C}}
\newcommand{\Pfhi}{\mathbf{P}_{H}^{I}}
\newcommand{\Pfhid}{\dot{\mathbf{P}}_{H}^{I}}
\newcommand{\Bf}{\mathbf{B}}
\newcommand{\Hf}{\mathbf{H}}
\newcommand{\Hfd}{\dot{\mathbf{H}}}
\newcommand{\Hfm}{\mathbf{H}^{-1}}
\newcommand{\Hfmd}{\dot{\mathbf{H}}^{-1}}
\newcommand{\Hfc}{\mathbf{H}_{C}}
\newcommand{\Hfcd}{\dot{\mathbf{H}}_{\mathbf{C}}}
\newcommand{\Hfcm}{\mathbf{H}_{\mathbf{C}}^{-1}}
\newcommand{\Hfcmd}{\dot{\mathbf{H}}_{\mathbf{C}}^{-1}}
\newcommand{\Xhxm}{ \big(\mathbf{X}^{\top}\mathbf{H}^{-1}\mathbf{X}\big)^{-1} }
\newcommand{\Xhxmd}{ 
	\big(\mathbf{X}^{\top}{(\dot{\mathbf{H}})}^{-1}\mathbf{X}\big)^{-1}}
\newcommand{\Xchcxm}{ \big(\mathbf{X}^{\top} \mathbf{C} \mathbf{H}_{\mathbf{C}}^{-1} 
	\mathbf{C} 
	\mathbf{X} \big)^{-1}  }
\newcommand{\Xchcxmd}{ \big( \mathbf{X}^{\top} \mathbf{C}\dot{H}_{\mathbf{C}}^{-1} 
\mathbf{C} 
	\mathbf{X} \big)^{-1}  }
\newcommand{\y}{\mathbf {y}   }
\newcommand{\ainv}{a^{-1}}
\newcommand{\binv}{b^{-1}}
\newcommand{\ninv}{n^{-1}}
\newcommand{\Sxh}{\hat{\mathbf{\Sigma}}_{\mathbf{X}}}
\newcommand{\Szh}{\hat{\mathbf{\Sigma}}_{\mathbf{Z}}}
\newcommand{\Szih}{\hat{\mathbf{\Sigma}}_{\mathbf{Z}_i}}
\newcommand{\Szizjh}{\hat{\mathbf{\Sigma}}_{\mathbf{Z}_i\mathbf{Z}_j}}
\newcommand{\Sxzh}{\hat{\mathbf{\Sigma}}_{\mathbf{XZ}}}
\newcommand{\Sxzih}{\hat{\mathbf{\Sigma}}_{\mathbf{XZ}_i}}
\newcommand{\Sbh}{\mathbf{\Sigma}_{\hat{\mathbf{\beta}}}}
\newcommand{\Shbh}{\dot{\mathbf{\Sigma}}_{\hat{\mathbf{\beta}}}}
\newcommand{\Sut}{\mathbf{\Sigma}_{\tilde{\mathbf{u}}}}
\newcommand{\Shut}{\dot{\mathbf{\Sigma}}_{\tilde{\mathbf{u}}}}
\newcommand{\Shuti}{\dot{\mathbf{\Sigma}}_{\tilde{\mathbf{u}}_i}}
\newcommand{\Shuiujt}{\dot{\mathbf{\Sigma}}_{\tilde{\mathbf{u}}_j\tilde{\mathbf{u}}_i}}
\newcommand{\eps}{\boldsymbol{\varepsilon} }
\newcommand{\be}{\boldsymbol{\beta} }
\newcommand{\uf}{\mathbf{u} }
\newcommand{\bh}{ \hat{\boldsymbol{\beta} } }
\newcommand{\bht}{\hat{\boldsymbol{\beta}}^{\top} }
\newcommand{\bd}{\dot{\boldsymbol{\beta}}}
\newcommand{\btd}{\dot{\boldsymbol{\beta}}^{\top}}
\newcommand{\ut}{\tilde{\mathbf{u}}}
\newcommand{\utt}{\tilde{\mathbf{u}}^{\top}}
\newcommand{\ud}{\dot{\mathbf{u}}}
\newcommand{\utd}{\dot{\mathbf{u}}^{\top}}
\newcommand{\muh}{\hat{\mu}}
\newcommand{\mud}{\dot{\mu}}
\newcommand{\Rbd}{\dot{S}_{\mathbf{X}}^2}
\newcommand{\Rbdi}{\dot{S}_{\mathbf{X}_{i}}^2}
\newcommand{\Rbdj}{\dot{S}_{\mathbf{X}_{j}}^2}
\newcommand{\Rb}{\hat{S}_{\mathbf{X}}^2}
\newcommand{\Rlmm}{\dot{R}^2}
\newcommand{\Rud}{\dot{S}_{\mathbf{Z}}^2}
\newcommand{\Rudi}{\dot{S}_{\mathbf{Z}_{i}}^2}
\newcommand{\Rbud}{\dot{S}_{\mathbf{X} \times \mathbf{Z}}^2}
\newcommand{\Ru}{\tilde{S}_{\mathbf{Z}}^2}
\newcommand{\Rbu}{\tilde{S}_{\mathbf{X} \times \mathbf{Z}}^2}
\newcommand{\Vf}{\mathbf{V}}
\newcommand{\phan}{\overset{\phantom{(\ref{orthogonal})}}{=}}
\newcommand{\pha}{\overset{\phantom{(\ref{Hinv})}}{=}}
\newcommand{\phannn}{\overset{\phantom{(\ref{orthogonal}),(\ref{orthogonal})}}{=}}
\newcommand{\phann}{\overset{\phantom{(\ref{orthogonal}),(\ref{Hinv})}}{=}}
\newcommand{\phandef}{\overset{\phantom{(\ref{orthogonal})}}{:=}}
\newcommand{\phadef}{\overset{\phantom{(\ref{Hinv})}}{:=}}
\newcommand{\iid}{i.\,i.\,d.\,}







\begin{abstract}
In the linear mixed model (LMM), the simultaneous assessment and comparison 
of dispersion relevance of explanatory variables associated with fixed and
random effects remains an important open practical problem. 
Based on the restricted maximum likelihood equations in the variance components 
form of the LMM, we prove a proper decomposition of the sum of squares of the 
dependent variable into unbiased estimators of interpretable estimands of 
explained variation. 
This result leads to a natural extension of the well-known 
adjusted coefficient of determination to the LMM. 
Further, we allocate the novel unbiased estimators of explained variation
to specific contributions of covariates associated with fixed and random effects 
within a single model fit. 
These parameter-wise explained variations constitute easily interpretable 
quantities, assessing dispersion relevance of covariates associated with both 
fixed and random effects on a common scale, thus allowing for a covariate 
ranking.

For illustration, we contrast the variation explained
by subjects and time in the longitudinal sleep deprivation study. 
By comparing the dispersion relevance of population characteristics and spatial 
levels,  we determine literacy as a major driver of income inequality in Burkina 
Faso. 
Finally, we develop a novel relevance plot to visualize the dispersion relevance 
of high-dimensional genomic markers in Arabidopsis thaliana.
\end{abstract}



\section{Introduction}
\label{sec1}

Decomposition of observed sum of squares into contributions of explanatory factors
and unexplained variation constitutes the basis of the analysis of variance (ANOVA). 
This concept is generalized in the linear model (LM), where the total sum of squares
is decomposed into explained and residual sum of squares, leading to the well-known 
coefficient of determination, $R^2$. 
It measures the proportion of variance in the response explained by the 
covariates and is widely applied to assess absolute goodness of model fit. 
The relevance or usefulness of individual covariates in the LM could be assessed, 
for example, by the associated estimated effect sizes or $p$-values, which, however, 
depend heavily on scaling and sample size, respectively \citep{Darlington1968}. 
Consequently, measures of dispersion relevance have come to the fore, such as
the proportion of explained to observed variation. 
Partitioning the explained variance into contributions of specific 
covariables is straightforward and unambiguous in the uncorrelated case. 
In observational studies, however, covariables are typically 
correlated, leading to a general debate about how to assess the usefulness or 
relevance of individual covariables in the LM \citep{Darlington1968}. 
On the one hand, the drop in model $R^2$ when a predictor is removed 
has been utilized  \citep{Darlington1968}.  
This approach, which is equivalent to the so-called semi-partial $R^2$ 
\citep{Groemping2015},
assesses how well the remaining covariates 
compensate the omitted variable in terms of model $R^2$
but does not lead to a proper partitioning of the model $R^2$ 
(i.\,e.\,semi-partial $R^2$'s do not add up to the model $R^2$). 
Hence, they are only useful in relation to each other
but do not lead to a meaningful interpretation in terms of individual explained
variation with respect to the model. 
On the other hand, fractions of $R^2$ from a single model fit
can be properly allocated to the covariates in the model, for example 
resulting in Engelhart's measure \citep{Engelhart1936} or the measure introduced
by \citet{Hoffmann1960} and formalized by \citet{Pratt1987}. 

The variance components form of the linear mixed model (LMM) 
extends the LM to handle correlated observations \citep{Searle1992}. 
Routine applications include modeling of hierarchical or multi-level data, 
e.\,g.\,in the presence of clustered, repeated or longitudinal measurements. 
Genome-wide association studies (GWAS) in which all genomic 
markers are fitted simultaneously as random slopes for regularization 
are another important application of the LMM in the 
variance components form \citep{Yang2010}. 
In this context, the combined explained variation of the covariates, 
representing genomic markers, associated with random effects is termed 
genomic variance of a trait \citep[e.g.][]{Yang2010}. 

Although arguably of equal importance as in the LM, it has not yet 
been agreed on a unified measure of explained variation in the LMM, amongst
others, because the LMM specifies not only the mean but also the covariance
structure via random effects \citep{Edwards2008} which introduces additional
uncertainty and heterogeneity. 
A general definition of a $R^2$ for both model aspects has been called mathematically 
unresolvable, hence explained variations have often been treated distinctively for both
model parts \citep{Edwards2008} and are inherently difficult to compare. 

Whereas \citet{Xu2003} defines three measures for the overall proportion of
explained variance via residuals (coefficients of alienation), 
\citet{Edwards2008} define an $R^2$ statistic for the 
fixed effects part of the marginal LMM based on an approximate $F$ statistic
while keeping the random effect structure constant. 
Complementary to that, \citet{Demidenko2012} investigate explained 
variation of random effects by treating them, conditionally, as fixed 
effects and then applying $F$-test statistics.
\citet{Nakagawa2013} propose the calculation of a marginal 
(variance explained by mean structure) and a conditional 
(variance explained by mean and covariance structure) $R^2$ which has
been extended to random slope and generalized LMMs
by \citet{Johnson2014} and \citet{Nakagawa2017}. 
Finally, generalized variance approaches projecting the (generalized) LMM covariance 
matrix to a single number, e.\,g.\,by determinant or semi-variance, have been used 
to calculate $R^2$ for the mean structure only, e.\,g.\,by \citet{Jaeger2019} and
\citet{Piepho2019, Piepho2023}. 
For a detailed review of $R^2$'s for mixed models we refer to \citet{Cantoni2021}. 

The usual results of a LMM fit include effect size estimates and $t$-statistics
for the fixed effects and estimates of the variance components corresponding to
the random effects. 
These results do not enable a parameter-wise comparison of the relevance of
covariables associated with fixed and random effects. 
The relevance of individual covariates associated with fixed effects has been 
judged using (semi-)partial F-statistics either in the marginal LMM 
\citep{Edwards2008} or using standardized general variance adjustments 
\citep{Jaeger2019}. 
For LMM's with identical covariance structure, \citet{Stoffel2021} propose
sequential comparisons of the variance of the fixed effects part of the linear
predictor in the full model versus nested models in which covariates associated
with fixed effects are iteratively dropped. 
These approaches effectively lead to semi-partial $R^2$'s for the fixed part of
the LMM. 
They require multiple model fits, do not add up to the model $R^2$ and do not constitute 
interpretable quantities in terms of model $R^2$ but only in relation to each 
other. 
Similar semi-partial approaches for the covariates associated with random effects
are not advisable, because dropping these coefficients implies an alteration of
the covariance structure of the models to be compared.
On the other hand, the intra-class correlation coefficient (ICC) 
\citep[e.g.][]{Roberts2011} and variance partition coefficients \citep[e.g.][]{Snijders2012}
have been used to determine the level-specific variations explained in random 
intercept and multi-level models, respectively.
In this context, computation of the explained variation accounted 
for by covariates associated with fixed effects is not intended.  
Further, an addition of covariates changes the magnitude of variance 
component estimates and possibly leads to negative level-specific 
explained variances.
Moreover, whereas the relevance of clusters in a random intercept 
model without fixed effects might be approximated by the ratio of variance
components (ICC), this assessment fails when fixed effects and higher order 
random effects such as random slopes are introduced. 

Hence, to the best of our knowledge, there is no method available which
properly partitions the model explained variation for covariates associated 
with fixed and random effects at the same time.
We believe that such an approach is the only viable option for parameter-wise 
and interpretable assessments of the simultaneous relevance of covariates 
associated with fixed and random effects in the LMM.

Our contribution is two-fold. 
Firstly, we prove a novel and proper decomposition of the observed variance 
in the LMM into unbiased estimators of meaningful and interpretable
estimands of explained variation.
We utilize our decomposition to derive an extension of the adjusted 
coefficient of determination to the LMM, and to attribute the total explained 
variation to the contributions of 
covariates associated with fixed and random effects. 
We emphasize that the concept of explained variation in the LMM is a 
combination of the variance of the data-generating process of covariates 
and the variance of the associated effects (zero for fixed but non-zero 
for random effects).
Our decomposition is proper because it includes the estimates 
of realizations of the random effects (empirical best linear unbiased prediction), 
the use of which has been advocated before \citep{Morris1983, Robinson1991}
such that the analysis is not solely based on population averages \citep{Xu2003}. 
Secondly, we partition the estimators of explained variation further and 
attribute components to individual parameters. 
This enables an unbiased, model-consistent and unified parameter-wise attribution of 
explained variation to single covariates associated
with fixed and random effects on the same scale within a single model fit. 
These contributions are conditional on the other variables in the model, 
i.\,e.\,we do not aim at selecting or removing model coefficients but to 
describe their dispersion relevance within the model at hand. 

Our approach is made readily 
available in a user-friendly $R$-package 
and its usefulness is illustrated in several real data sets to 
emphasize the broad applicability. 
For instance, in the sleep deprivation study analysed by \citet{lme4},
we show that about 60\% of the variation in the reaction time of subjects is
explained by the time of sleep deprivation, with approximately
equal contributions of time associated with fixed and random slopes, while
20\% are explained by the subjects themselves. 
\citet{Graeb2011} investigate inequality in Burkina Faso by using
the ICC to decompose the variability of household expenditure in a multilevel 
spatial model.  
We illustrate that the dispersion relevance of spatial levels, such as communities,
provinces and regions can be simultaneously assessed together with the dispersion
relevance of household and community characteristics. 
Amongst others, we conclude that covariates related to literacy explain about 
$23\%$ of the variance of household expenditure (proxying income inequality), 
while the spatial levels explain about $10\%$ in total.   
Finally, in the context of GWAS data on Arabidopsis thaliana, 
we introduce a novel relevance plot for the genomic 
markers on the whole genome complementing the well-known Manhattan plot which 
neglects confounding and assesses significance but not relevance. 


\section{Preliminaries}

\subsection{Linear Mixed Model and Notation}

We consider the linear mixed model
\begin{equation}\label{LMM}
	\  \y =  \mu \ones + \Xf \be + \sum_{i=1}^{m} \Zf_i \uf_i + \eps,
	\quad \eps \sim \No(\mathbf{0}, \se\Id),
	\quad \uf_i\sim\No(\mathbf{0}, \sui \mathbf{I}_{p_i})
\end{equation}
in the traditional variance components form 
\citep[\S 6, \S 7]{Searle1992} 
where $\y$ is the vector of quantitative observations, $\mu\in\R$ is an 
intercept, $\ones$ is the 
column vector of ones of length $n$, $\Id$ denotes the $n$-dimensional 
identity matrix, and $\se >0$ and $\sui \geq 0$, $i=1,...,m$.
The design matrix for the $k$ fixed effects $\be$ is given by 
$\Xf\in \R^{n\times k}$. 
The matrices $\Zf_i \in \R^{n \times p_i}$ $(i=1,...,m)$ code covariates 
corresponding to the random effect vectors $\uf_i$.
The random effects (including the residuals) are pairwise uncorrelated. 
By specifying $\Zf_i$, model (\ref{LMM}) typically allows to model 
hierarchical data (clustered data, repeated measurements, 
longitudinal data) as well as high-dimensional GWAS data. 
Let the random vector 
$\uf=(\uf_1^{\top},\ldots, \uf_m^{\top})^{\top}\sim \No(\mathbf{0}, \Df)$
be of dimension $p$, 
$\Df= \diag(\suone \mathbf{I}_{p_1},\ldots, \sur \mathbf{I}_{p_m})$,
and $\Zf =(\Zf_1,\ldots, \Zf_m)$ denote an $n\times p$ design matrix. 
Model (\ref{LMM}) has variance-covariance structure
$\Vf = \Cov(\y)=  \Zf   \Df  \Zf^{\top} + \se\Id = \sum_{i=1}^{m}\sui \Zf_i
\Zf^{\top}_i + \se\Id$,
such that $\y\sim\No(\ones \mu +  \Xf \be,  \Vf )$. 
Similar to \citet[\S 6]{Searle1992}, we use $\Gf = \Df/ \se$. 
Let $\Hf= \Vf / \se=  \Id + \Zf \Gf \Zft$ and $\Hfc= \Id + \Cf \Zf \Gf \Zft \Cf$, 
which originates from centring the columns of $\Zf$ by the centering matrix 
$\Cf= \Id- n^{-1}\ones\onest$. 
We assume that $\tilde{\Xf}= (\ones, \Xf)$ has full column rank and that the 
inverses of $\Hf$, $\tilde{\Xf}^{\top}\Hfm\tilde{\Xf}$, $\Hfc$ and 
$\Xf^{\top}\Cf \Hfcm \Cf \Xf$ exist. \newline
The following lemma introduces alternative expressions for the well-known best linear 
unbiased estimator (BLUE) $\bh$ and predictor (BLUP) $\ut$ which illustrate their 
invariance to centring the data $\y$, $\Xf$, and $\Zf$.  
The invariance also holds for the respective covariance matrices
of $\bh$ and $\ut$. 
The invariance property is important for explained variations and in
particular striking because centring $\Zf$ automatically changes the 
covariance structure $\Vf$ 
of the outcome. 
The proof can be found in the Supplemental Material. 
\begin{lemma}\label{blueblup}
	The best linear unbiased estimator $\bh$ in model (\ref{LMM})
	can be represented by $\Xchcxm\Xft \Cf \Hfcm \Cf \y$ with covariance 
	matrix $\Sbh= \se \Xchcxm$. 
	The best linear unbiased predictor $\ut$ of the random effects can be
	expressed as $\Gf\Zft\Cf\Hfcm\Cf\big(\y - \Xf\bh\big)$ with covariance 
	matrix $\Sut=\se\Gf\Zft\Cf\Hfcm\Big\{\Hfc-
	\Cf\Xf\Xchcxm\Xft\Cf\Big\}\Hfcm\Cf\Zf\Gf.$
\end{lemma}
Typically, estimation of variance components is based on restricted maximum 
likelihood (REML). 
When we replace the variance components by their estimates we mark the 
resulting quantities by a dot. 
In case that the variable already has a superscript, we drop the first 
superscript for notational convenience.

\subsection{Explained Variation and $R^2$ in the Linear Model}
\label{sec: LM_R2}

A special case of model (\ref{LMM}) is the homoscedastic LM
\begin{equation}\label{statmod}
	\y = \mu\ones + \Xf\be + \eps,\quad \eps\sim\No(0, \se\Id).
\end{equation}
The coefficient of determination $ R^2= 1 - \RSS / \TSS$
is the most popular measure of goodness-of-fit in model (\ref{statmod}),
where $\TSS= \y^{\top}\Cf \y$ is the centred total sum of squares, 
$\RSS= \hat{\eps}^{\top}\hat{\eps}$ is the residual sum of 
squares, and $\hat{\eps}$ are the estimated 
residuals, respectively.
The centred explained sum of squares is $ \ESS = \hat{\y}^{\top}\Cf \hat{\y}$ 
where $\hat{\y}$ are the fitted values. 
The intercept in model (\ref{statmod}) allows for the proper decomposition 
$\TSS = \ESS + \RSS$, which leads to $R^2= \ESS / \TSS$, 
i.\,e.\,the coefficient of determination measures the proportion of 
total variation that can be explained by the covariates after a least 
squares fit, see e.\,g.\,\citet{Kvalseth1985}. 
The $R^2$ tends to overestimate the proportion of the explained 
variation which is why the adjusted coefficient of determination
$\bar{R}^2= 1 - (1 - R^2) (n-1)/(n-k-1)= 1 - \seh/\syh$ has been 
introduced, adjusting for the degrees of freedom \citep{Kvalseth1985}. 
The expectation of the sample variability in the observations, 
$\syh= \TSS / (n-1)$, is given by
\begin{equation} \label{ELM}
	\E\big[\syh \big] = \frac{1}{n-1}\tr\Big( \Cf \big(\Cov(\y) - 
	\E[\y]\E[\y]^{\top} \big) \Big)= 
	\se + \be^{\top} \Sxh \be, \quad \Sxh = \frac{1}{n-1}\Xft \Cf \Xf.  
\end{equation}
It comprises residual variation and variation due to the covariates
associated with the fixed effects, $\be^{\top}\Sxh\be$, which is induced by the 
empirical variation of the explanatory variables $\Xf$. 
An unbiased estimator for the estimand, the sample variance $\be^{\top}\Sxh\be$ of 
$\Xf\be$, is given by
$\Rbd= \bht\Sxh\bh - \tr(\Sxh\Shbh)$.
The estimator $\Rbd$ measures the adjusted empirical variation explained 
by the covariates and their empirical pairwise covariances.
In combination with $\seh= \RSS / (n-k-1)$ and $\TSS = \ESS + \RSS$
we obtain the proper decomposition
	$ \syh= \Rbd + \seh$
of the centred $\TSS$ into the unbiased estimator of the
variability of covariates associated with the fixed effects and an unbiased 
estimator for the residual variation. 
Thus, an alternative expression for $\bar{R}^2$ in model (\ref{statmod})
follows as $\bar{R}^2  = \Rbd/(\Rbd + \seh)$. 
Of note, $\Rbd$ and with it $\bar{R}^2$ potentially produce negative estimates
($\bar{R}^2 \in [-k/(n-k-1), 1]$). 
This can be the case, for instance, if $R^2=0$ and $k > 0$ and if $\bh = 0$ 
with $\Sbh \neq 0$, i.\,e.\,covariates have low/null explanatory potential but 
estimation uncertainty.

\subsection{Components of Explained Variation in the Linear Model}
\label{sec: part}

A general partitioning of the variation of the (stochastic) LM is given by 
$  \ \Var(Y) = \sum_{j=1}^k \beta_j^2 \Var(X_j) + 
2 \sum_{j = 1}^k \sum_{i < j}\beta_i\beta_j \Cov(X_i, X_j) + \Var(\varepsilon)$, 
see, e.\,g.\,, \citet{Engelhart1936}.
The terms $\beta_j^2\Var(X_j)$, $j = 1, \ldots, k$, have been termed \textit{direct}
contributions of the respective variables, and the covariance terms have been called
\textit{joint} contributions of pairs of variables to the
variance of the outcome. 
While the sum of all these contributions results in $R^2$, the highly debated
issue is how to distribute the \textit{joint} contributions to obtain \textit{total}
contributions of each parameter.  
Assigning parts of the \textit{joint} distribution weighted by the proportion of
the corresponding variables' \textit{direct} contributions possibly constitutes the
first proposal to share the covariance terms \citep{Engelhart1936}.
One obvious, straightforward and data-independent approach is to assign each 
variable one part of its corresponding covariance terms implying equal shares 
as introduced by \citet{Hoffmann1960}. \citet{Pratt1987} justified the 
measure as the only one satisfying a set of desirable criteria. 
The so-obtained \textit{total} contributions can become negative, which has been 
criticised \citep{Groemping2015}.
However, variables might decrease the variance of the observations by a 
strong negative correlation with other variables, hence negative values are not
necessarily unjustified \citep{Pratt1987}. 
More complicated assignments including sequential sum of squares in combination with
iterative weighting (LMG) and data-dependent approaches such as the proportional
marginal variance decomposition (PMVD) have been reviewed in \citet{Groemping2007}. 
Both LMG and PMVD derive from game theory and are often only heuristicly 
discussed because the interpretation of the estimated quantities (estimand) is 
unclear \citep{Groemping2015}. 
Additionally, LMG depends on the ordering of the entering into the regression 
equation and is computationally challenging \citep{Groemping2015}.
Contrary to that, the Hoffman-Pratt measure is argued to be the only theory-based 
measure with simple means of estimation and straightforward interpretation 
\citep{Thomas2017}.

\section{Explained Variation in the Linear Mixed Model}
\label{subsec: LMM}

\subsection{Main Result - Decomposition}

The lemma below introduces unbiased estimators and predictors for empirical 
explained variances which constitute the estimands of interest in this paper, 
see also \citet{Schreck2018} and an earlier preprint of this manuscript. 
\begin{lemma}\label{empvars}
	Based on model (\ref{LMM}) and the quantities introduced in Lemma 
	\ref{blueblup}, the estimator $\Rb = \bht\Sxh\bh-\tr\big(\Sxh\Sbh\big)$
	is unbiased for the estimand for the explained variation of the coefficients 
	associated with fixed effects, the sample variance $\sigma_f^2 = \be^{\top}\Sxh\be$ 
	of $(\Xf\be)_i$ $(i=1,\ldots, n)$. 
	The predictor
	$\Ru =  \utt \Szh\ut-\tr\big(\Szh\Sut\big) + \tr\big(\Df\Szh\big)$
	is unbiased for the estimand for the explained variation of the coefficients 
	associated with random effects, the sample variance $\sigma_r^2 = \uf^{\top}\Szh \uf$ 
	of $(\Zf\uf)_i$ $(i=1,\ldots, n)$. 
	The predictor $\Rbu= 2\bht\Sxzh \ut $ is unbiased for the estimand for the 
	explained variation of the covariance induced by possible non-orthogonality of 
	the columns of $\Cf\Xf$ and $\Cf\Zf$, the sample covariance $2\be^{\top}\Sxzh \uf$.
	The quantities $\Rb$, $\Ru$, and $\Rbu$ are invariant to
	centring the data $\y$, $\Xf$, and $\Zf$. 
\end{lemma}
Later, \citet{Piepho2023} derived an estimator for the average semisquared bias 
for the mean structure which equals $\Rb$ in the case of the variance components
form of the LMM. 

The following theorem states
the decomposition of the centred total sum of squares. 
\begin{theorem}\label{main}
	Given restricted maximum likelihood estimators $\seh$ for the residual 
	variance $\se$ and $\suih$ for the variance components $\sui$
	$(i=1,\ldots, m)$ in model (\ref{LMM}), we find
	\begin{equation}\label{vardecomp}
		\syh= \Rbd + \Rud + \Rbud + \seh.
	\end{equation}
\end{theorem}
Decomposition $ \syh= \Rbd + \seh$ and the deduction and interpretation of $\Rbd$
follow in the special case of the LM as presented in Section \ref{sec: LM_R2}.
In the special case of a random effects model with only one 
random effect vector and no fixed effects but the intercept, decomposition 
(\ref{vardecomp}) reduces to $ \syh = \Rud + \seh$, 
which has been observed, but not 
been proven, by \citet{Schreck2019} in several genomic applications. 
In that context, $\Rud$ has been termed empirical best predictor of the additive 
genomic variance \citep{Schreck2019}. \newline
Due to the explicit nature of the estimand of the different variance contributions
expressed in Lemma \ref{empvars}, their interpretation is intuitive. 
In particular, the explained variations (estimands) originate from sample variations of 
the covariate information, no matter whether their associated effects are 
fixed or random. 
This reflects that the data-generating processes of $\Xf$ and $\Zf$ are vital drivers of
the underlying variability. 
While for the mean structure variability is induced only through the data-generating 
process of the covariate corresponding to fixed effects, in the case of random 
effects both the associated data-generating process and the effects themselves 
induce variability, which is clearly reflected within $\Rud$.
The proper decomposition of the variability of the observations in (\ref{vardecomp}) was
made possible by the inclusion of the estimates of the realizations of the random effects. 
The latter are viewed as parameters to be estimated, reflecting an empirical best prediction 
approach (or parametric empirical Bayes, if the prior is Gaussian) as advocated, for 
instance, in \citet{Morris1983}, \citet{Robinson1991}, \citet{Xu2003}.
Consequently, the corresponding explained variation is not solely concerned with the 
contribution of \textit{population averages} but also of \textit{data set specific} 
realizations, as further discussed in the next section. 
Possible correlations between the data-generating processes of $\Xf$ and $\Zf$ 
are captured by $\Rbud$.
Further interpretations follow in the next sections and in
the applications. 

\subsection{Population and Data set Specific Explained Variation}
\label{sec: dsd}

The empirical predictor $\Rud$ for the explained variation associated with the 
random effects, see Lemma \ref{empvars} and Equation (\ref{vardecomp}),
can be further decomposed as
\begin{equation*}
	\Ru = \Ru{}^{p} + \Ru{}^{d}, \quad \Ru{}^{p} = \tr\big(\Df\Szh\big), 
	\quad \Ru{}^{d} = \utt \Szh\ut -\tr\big(\Szh\Sut\big),
\end{equation*}
into what we term \textit{population-specific} (indexed by $p$) and 
\textit{data-specific} (indexed by $d$) explained variation. 
The terminology and interpretation are based on the following reasoning. 
\newline
Similar to equation (\ref{ELM}), the expectation of the quadratic 
form of the observations is
\begin{equation}\label{ELMM}
	\E\big[\syh \big] = \frac{1}{n-1}\tr\Big( \Cf \big(\Cov(\y) - 
	\E[\y]\E[\y]^{\top} \big) \Big)= 
	\se + \be^{\top} \Sxh \be + \tr(\Df \Szh). 
\end{equation}
Whereas bias-reduced estimators for $\se$ and $\be^{\top} \Sxh \be$ 
are directly reflected by $\seh$ and $\Rbd$ in decomposition (\ref{vardecomp}),
the relationship of $\Rud$ and $\Rbud$ with Equation (\ref{ELMM}) might not
be obvious. 
From Lemma \ref{empvars}, recall that the sample variance 
$\sigma_r^2 = \uf^{\top}\Zft \Cf \Zf \uf / (n-1)$
of the random vector $\Zf \uf$
is itself a random variable in the context of model (\ref{LMM}). 
The unconditional expectation
$\E[ \sigma_r^2 ] = \tr(\Df \Szh)$
represents the last term in Equation (\ref{ELMM}).
Similar to the best linear unbiased predictor $\E[\uf \,|\,\y]$ of the random 
effects $\uf$, the conditional expectation is 
\begin{align*}
	\E\big[ \sigma_r^2 \,|\,\y \big] = 
	\tr\big( \E \big[\uf \uf^{\top} \,|\, \y \big] \Szh \big)& = 
	\E \big[\uf^{\top} \,|\, \y \big] \Szh \E \big[\uf \,|\, \y \big] + 
	\tr(\Df \Szh) - \tr( \Sut \Szh),  
\end{align*}
where the last equality holds because of the variance components form of
model (\ref{LMM}).

The empirical predictor $\Rud$ directly estimates $ \E\big[ \sigma_r^2 \,|\,\y \big]$.
This conditional expectation includes the unconditional expectation $\tr(\Df \Szh)$
and the BLUP's of the random effects as a measure of realized effect size, 
similarly to the fixed effects estimates in the variance explained by covariates 
associated with fixed effects. 
Conditioning on the data, it utilizes the information provided by the data $\y$.
From a Bayesian perspective, $\Rud$ can be interpreted as the posterior
expectation of explained variation of covariates associated with random
effects, also called parametric empirical Bayes estimator \citep{Morris1983}. 
Compared to the unconditional expectation (or the prior expectation 
$\tr(\Df \Szh)$), the conditional expectation has superior 
properties as an estimator of the realizations of a random variable 
\citep[e.\,g.\,]{Searle1992}. 
Consequently, we prefer $\Rud$ to $\tr(\dot{\Df} \Szh)$ as an estimator of
the realization of $\sigma_r^2$. 
The terminology \textit{population-specific} explained variation for $\Ru{}^{p}$
originates from the derivation as the unconditional expectation of $\sigma_r^2$
and the implied interpretation as the mean of the hypothetical population.
The deviation of the conditional from the unconditional expectation, $\Ru{}^{d}$, 
has expectation $0$ with respect to the population.
Hence, we call $\Ru{}^{d}$ \textit{data-specific} explained 
variation.  

For instance in longitudinal studies with fixed and random time effect, 
$\Rb$ can be viewed as the explained variation associated
with the covariates associated with the mean effect of time. 
The corresponding explained variation associated with random slopes can be split
into $\Ru{}^{p}$ which resembles the average of the individual variation
around the fixed slope in the hypothetical population. 
The additional deviations of the variation in the underlying data set 
given the information of the sampling population are accounted for by
$\Ru{}^{d}$.

In Section \ref{sec: real}, we illustrate that the total \textit{population-specific}
variance often contributes the largest part to the explained 
variation associated with random effects and their respective covariates. 
The \textit{data-specific} explained variation seems to contribute strongly
if the covariates associated with random effects are highly correlated and is the 
main driver in our novel parameter-wise dispersion relevance profile, making 
individual covariates more distinguishable. 
These individual contributions are, however, not restricted to be positive, and
often do not contribute largely in sum to the overall explained variation if the
covariates are uncorrelated. 

The empirical covariance $\be^{\top}\Sxzh \uf$ between the vector $\Xf \be$ and 
$\Zf \uf$ always has unconditional expectation (\textit{population-specific} explained 
variation) $0$ and conditional expectation (\textit{data-specific} explained 
variation) $\Rbu$.
For this reason, the cross-terms do not contribute in Equation (\ref{ELMM})
but do so in decomposition (\ref{vardecomp}).

\subsection{Coefficients of Determination}\label{partialR2_LMM}

The explained variations defined in Lemma \ref{empvars} depend on the amount of 
observed variation. 
Coefficients of determination are often defined as explained variance scaled with 
respect to the total observed variation or the residual variance of a reference model, 
respectively, in order to enable comparability of variation explained in different
settings.  
In Section \ref{sec: LM_R2}, the reference model is the intercept-only model where the
residual variance coincides with $\syh$, leading to a consistent scaling.  
In general, however, there need not be a unique reference model, leading to a 
variety of potential $R^2$'s. 

In accordance to $\bar{R}^2$ in Section \ref{sec: LM_R2}, we propose to 
utilize decomposition (\ref{vardecomp}) to define
\begin{equation}\label{R2lmm}
	\Rlmm = \frac{1}{\syh}\Big(\Rbd + \Rud + \Rbud\Big) =
	\frac{\Rbd + \Rud +  \Rbud}{\Rbd + \Rud + 
		\Rbud+\seh} = 1 - \frac{\seh}{\syh}
\end{equation}
as the \textit{adjusted} coefficient of determination for the LMM (\ref{LMM}). 
The terminology refers to the fact that (\ref{R2lmm})
reduces to $\bar{R}^2$, see Section \ref{sec: LM_R2}, and uses bias-corrected 
estimates by adjusting for estimation uncertainty of the effects sizes.  
The total observed variance or the residual variance in the model with only
a fixed intercept (reference model), implying all observations to be 
uncorrelated, is used for scaling. 
The $\Rlmm$ has, by definition, the interpretation as proportion of explained variance, 
and as the additive inverse to the coefficient of alienation 
(proportion of unexplained variance). 
Therefore, (\ref{R2lmm}) is in line with the $r_0^2$ introduced in \citet{Xu2003}, 
which also reduces to the adjusted coefficient of determination in the LM in the 
case that REML estimation is used \citep{Xu2003}. 
Our main focus, however, is on explicitly investigating and interpreting the 
explained variation in detail. 

The $\Rlmm$ is an analytic function of the estimates returned after 
a model fit with REML, making its computation straightforward. 
Due to the explicit and additive nature of decomposition (\ref{vardecomp}), 
we can easily define proportions of variance explained by the coefficients associated
with fixed or with random effects. 
Due to Lemma \ref{empvars}, the quantity $\Rlmm$ is dimensionless, 
i.\,e.\,it is invariant to scaling and centring the data 
$\y$, $\Xf$, and $\Zf$. 

Alternatively, based on the further decomposition given in Section
\ref{sec: dsd}, we can define a \emph{population coefficient of determination}
\begin{equation}\label{R2lmmp}
	\Rlmm_{p} = \frac{\Rbd + \sum_{i=1}^m\suih\tr\big(\Szih\big)  }
	{\Rbd + \sum_{i=1}^m\suih\tr\big(\Szih\big) + \seh}
\end{equation}
which constitutes the relative population average explained variation. 
The scaling has been chosen to resemble the unconditional expectation, see
Section \ref{sec: dsd}.
However, the implied null model is not obvious. 
As mentioned above, the estimator for the average semisquared bias 
in \citet{Piepho2023} for the mean structure equals $\Rbd$ in the variance 
components form of the LMM.
Naturally, the coefficient of determination introduced in \citet{Piepho2023} 
is identical to $\Rlmm_{p}$ in (\ref{R2lmmp}) already introduced in an earlier
preprint.

The widely employed marginal ($R^2_{m}$) and conditional 
($R^2_{c}$) coefficients of determination by \citet{Nakagawa2013} and 
\citet{Johnson2014} are given by
\begin{equation}\label{Rnaka}
	R^2_{m} = \frac{\hat{\sigma}_{f}^2 }{ \hat{\sigma}_{f}^2 + \hat{\sigma}^2_l + 
		\seh }, \quad 
	R^2_{c} = \frac{\hat{\sigma}_{f}^2 + \hat{\sigma}^2_l }{ \hat{\sigma}_{f}^2 + 
		\hat{\sigma}^2_l + \seh }
\end{equation}
where $\sigma^2_l = \tr(\Zf\Df \Zf^{\top} )/n $ is termed the
mean random effect variance, resembling the unconditional estimate in Section 
\ref{sec: dsd}.  
The term $\hat{ \sigma }_{f}^2 = \bht \Sxh \bh / n$ is used as an estimator
for $\sigma_{f}^2 = \Var(X \be )$.
This makes $R^2_c$ structurally similar to the population coefficient $\Rlmm_{p}$,
but not $\Rlmm$. 
In the proof of Lemma \ref{empvars}, we show that the estimator 
$\hat{ \sigma }_{f}^2$ for the fixed effect 
variance used in \citet{Nakagawa2013} is biased. 
The main differences of (\ref{Rnaka}) and (\ref{R2lmm})
are that we use a bias-corrected, variance adjusted estimator, $\Rbd$, for
the explained variance of the mean structure and that we include not only
\textit{population averages} for the covariance structure, but also 
\textit{data set specific} deviations. 
These deviations might be present for the coefficients of random effects but
also in terms of correlations between the coefficients associated with fixed and
random effects (given by $\Rbud$). 
The \textit{data set specific} contributions are in particular vital for parameter-wise 
relevance profiles, as illustrated, for instance, in Figure 
\ref{fig: arab5} and described in detail in the next section. 

In GWAS, the simultaneous effect of markers is typically modeled using the
variance components form of a random effect model, 
$\y = \mu \ones + \Zf \uf + \varepsilon = \mu \ones + \mathbf{g} + \varepsilon, 
\quad \mathbf{g} \sim \No(0, \su \Zf\Zft)$,
where $\Zf$ denotes the design matrix of the genomic markers \citep[e.g.][]{Yang2010}. 
The most popular approach for the estimation of the relative genomic variance
has been introduced by \citet{Yang2010} who use the estimator 
$\hat{\sigma}_{g}^2/\hat{\sigma}_y^2:= \suh\tr(\Szh)/\hat{\sigma}_y^2$. 
It is straightforward to see that this resembles the unconditional expectation, 
see Section \ref{sec: dsd}, and that $\suh\tr(\Szh)$ is similar to $\hat{\sigma}^2_l$ 
in Equation (\ref{Rnaka}). 
This estimator is not in accordance to classical quantitative genetics theory, see
\citet{Schreck2019} and the references therein. 
The estimator $\Rud$, however, is in accordance to 
quantitative genetics theory and in particular includes the contribution of 
linkage disequilibrium \citep{Schreck2019}. 
Consequently, $\Rud$ can be viewed as an improved estimator of the genomic
variance while being defined in the context of model (\ref{LMM}),
thus allowing for fixed effects and an arbitrary number of independent random effects. 
Heritability estimates, i.\,e.\,proportion of genomic variance, can be obtained 
by coefficients of determination based on (\ref{R2lmm}).

\subsection{Components of Explained Variation}
\label{sec: part2}

We believe that partitioning the explained variation is the best 
and possibly only option for a simultaneous assessment of parameter-wise relevance 
for covariates associated with fixed and random effects in the LMM.  
To this end, we utilize a similar reasoning as \citet{Hoffmann1960} and \citet{Pratt1987}, 
see Section \ref{sec: part}, to further partition the explained variations 
in decomposition (\ref{vardecomp}) to parameter-wise 
components. 
We assign equal shares of the \textit{joint} contribution to
the \textit{direct} contribution of parameters to obtain their
\textit{total} contribution. 

We partition the bias-corrected $\Rbd$, see Lemma \ref{empvars}, into the 
$k$ contributions of the covariates associated with fixed effects
$\Rbd= \sum_{j=1}^{k} \Big\{(\Sxh)_{jj} \Big[ \bh_j^2 - (\Sbh)_{jj} \Big]
+ \sum_{i\neq j}(\Sxh)_{ij} \Big[ \bh_i\bh_j - (\Sbh)_{ij} \Big] \Big\}$.
We attribute the \textit{data set specific} term $\Rbud$ to each of the $k$ 
fixed effects and to each of the $m$ random effect vectors by summing over the 
contributions of the covariances with the random and fixed effects, respectively
$\Rbud= 
\sum_{j=1}^{k} \bd_j \sum_{i=1}^{m} (\Sxzih \ud_i)_j + 
\sum_{i=1}^{m}\btd\Sxzih \ud_i$.
In this way, a certain contribution of data set specific explained variance can be 
attributed to each covariate, whether associated with a fixed or random 
effect. 
To sum up, we assign the $j$-th covariate associated with a fixed effect 
the share
\begin{equation}\label{Rxexplicit}
	\Rbdj= \Big\{(\Sxh)_{jj} \Big[ \bh_j^2 - (\Sbh)_{jj} \Big]
	+ \sum_{i\neq j}(\Sxh)_{ij} \Big[ \bh_i\bh_j - (\Sbh)_{ij} \Big] \Big\} + 
	\bd_j \sum_{i=1}^{m}(\Sxzih \ud_i)_j. 
\end{equation}
including \textit{data set specific} contributions.

It is straightforward to partition the \textit{population} explained 
variance $\sum_{i=1}^{m}\suih\tr\big(\Szih\big)$ to specific random effects. 
The \textit{data set specific} part can be partitioned using
$ 
\sum_{i=1}^{m}\Big\{ \tr(\Szih[ \ud_i \ud_i^{\top} - \Shuti] ) + 
\sum_{j\neq i}  \tr(\Szih[ \ud_i \ud_j^{\top} - \Shuiujt] )    \Big\}$.
Similar to (\ref{Rxexplicit}), we define the contribution of the $i$-th covariate
associated with random effects to the explained variance including data set
specific deviations as 
\begin{equation}\label{Rzexplicit}
	\Rudi = \suih\tr\big(\Szih\big) + 
	\tr(\Szih[ \ud_i \ud_i^{\top} - \Shuti] ) + 
	\sum_{j\neq i}  \tr(\Szih[ \ud_i \ud_j^{\top} - \Shuiujt] ) + \btd\Sxzih \ud_i.
\end{equation}

For instance in Table \ref{tab:income}, we illustrate how the parameter-wise 
contributions (\ref{Rxexplicit}) and (\ref{Rzexplicit}) 
lead to a simultaneous dispersion relevance assessment of all model covariates 
associated with fixed effects and random effects on the same scale. 
Relative parameter-wise shares of explained variations are obtained by 
division as in equation (\ref{R2lmm}). 

By similar arguments, the contribution of the covariates $\Zf_i$, 
see Equation (\ref{Rzexplicit}), can be even further partitioned to each 
column of $\Zf_i$. 
This can be applied, for example, to obtain relevance profiles for individuals
genomic markers which share the same random effect vector, as exemplified in 
Section \ref{subsec: gwas} and visualized in Figure  \ref{fig: arab5}.

\section{Real Data Applications}\label{sec: real}

\subsection{Hierarchical Data - Sleep Deprivation Study}
\label{subsec: sleep}

\citet{lme4} investigate reaction time per day for 
subjects in a longitudinal sleep deprivation study.
For each of the $18$ subjects, the response variable \textit{Reaction} 
defined as the average reaction times in milliseconds is collected on the 
same $10$ days, implying a balanced design.  
The amount of sleep is unrestricted on day $0$ and restricted to $3$ hours 
afterwards.
By means of a likelihood ratio test and difference in deviance, \citet{lme4} 
study whether a random slope for \textit{time} should be added to the 
basic model including \textit{time} as a covariate with fixed effect and a 
random subject-specific intercept (\textit{ID}).
They conclude that "adding a linear [random] effect of time uncorrelated 
with the [random] intercept leads to an enormous and significant drop in 
deviance". 
While this may be convincing from a modeling perspective, the magnitude of 
a deviance drop is inherently difficult to interpret.

We fit model (\ref{LMM}) and use Equation (\ref{vardecomp}) to additively decompose 
the observed variance of \textit{Reaction} ($\syh=3172.93$) into the estimate 
$\Rbd=888.78$ ($\approx 28\%$) for the adjusted explained variation of \textit{time} 
associated with a fixed effect, an estimate of the unbiased predictor $\Rud=1630.57$ 
($\approx 51\%$) for the explained variation of \textit{time} and \textit{ID} associated 
with random effects and the REML estimate $\seh=653.58$ ($\approx 21\%$) for the residual 
variance. 
The estimate of 
$\Rbud$ is negligible here ($< 0.01$), similar to 
$\Rud{}^{d}$. 
This yields 
$\Rlmm\approx 0.79$ 
which is very close to the population coefficient
$\Rlmm_{p}$. 

\bigskip

\begin{table}[ht]
	\singlespace
	\captionsetup{font=small}
	\small
	\caption{\doublespacing Sleep Deprivation Study. 
		$\dot{R}^2_{(\cdot)}$ denotes parameter-wise contributions to dispersion relevance, 
		as defined in equation and (\ref{Rxexplicit}) and (\ref{Rzexplicit}), 
		relative to $\syh$.
		We include $95\%$ parametric bootstrap confidence intervals.
		\label{tab:sleepstudy}}
	{\tabcolsep=4.25pt
		\begin{tabular}{@{}llrrrrrr@{}}
				Effect Type & Covariable & Estimate & Std. Error & Pr($>$ $|$t$|$) & 
				$\dot{R}^2_{(\cdot)}$ (in \%) & 2.5\% & 97.5\% \\
			\hline
			Fixed & (Intercept) & 251.41 & 6.89 & $<$ 0.001 & -  & - & - \\ 
			& \textit{time} & 10.47 & 1.56 & $<$ 0.001 & 28.01 & 15.43 & 44.48 \\ 
			Random & \textit{ID} & $\dot{\sigma}_{u_1}^2/\syh$~ = ~0.20 & - & - & 19.53 &  4.56 & 
			33.77 \\ 
			& \textit{time} & $\dot{\sigma}_{u_2}^2/\syh$~ = ~0.01 & - & - & 31.86 & 14.04 & 
			48.87\\ 
			Residual &          & $\dot{\sigma}_{\varepsilon}^2/\syh$~ = ~0.21 & - & - & 20.60 & 14.10 
			& 32.04 \\ 
	\end{tabular}}
\end{table}

In Table \ref{tab:sleepstudy} we present the relative
contributions of the individual covariates to the explained variance. 
This enables us to compare the relevance of all covariates on the same 
scale, which is not possible by using only the standard output including
effect estimates, standard errors, p-values and variance component estimates 
(relative to $\syh$). 
\textit{Time} as a covariate associated with a random slope explains 
about $32\%$ of the observed variation in \textit{Reaction}, substantiating the
result from the likelihood ratio test.
\textit{ID} contributes similarly as
the residual variance (about $20\%$ each) to the observed variance.
Notably, when comparing only the variance components of the covariates associated
with random effects, the variance of \textit{ID} is about $20$ times the variance
of \textit{time}. 
Clearly, this naive approach ignoring the variance $\Szh$ of the data-generating 
process of \textit{time} is not advisable. 
\textit{Time} as a covariate associated with a fixed effect contributes 
similarly (about $28\%$) to the covariate \textit{time} associated with a random slope.  
Thus, variability in \textit{Reaction} explained by \textit{time}
associated with both a fixed and random effect amounts to a total of about
60\%, which is considerably more than the subject-specific explained variation. 

In sum, explained variances are more intuitively 
interpretable than differences in deviance, capturing relevance of the 
covariates rather than statistical significance in sequential testing. 
The novel importance ranking provides complementary 
information to the standard LMM model output by assessing the relevance 
of all covariates on a common scale.

\subsection{Multilevel Data - Income Inequality in Burkina Faso}
\label{subsec: bf}

\citet{Graeb2011} investigate income inequality in Burkina Faso in the year 1994. 
At that time, the country was organized in 13 regions, 45 provinces and 434 
communities. 
Household data were drawn from nationwide household surveys covering $n = 8374$ 
households. 
The overall aim of \citet{Graeb2011}'s study was to ``[\ldots] provide a 
decomposition of the variance of observed living standards into components due 
to factors that vary within and factors that vary between spatial units'' 
to potentially ``[\ldots] help to target poverty reduction policies more 
effectively''. 
Thus, the authors were not primarily interested in estimating regression coefficients, 
but in decomposing the variance of the logarithm of \textit{household expenditure}
as a proxy for income inequality.
Consequently, the components of explained variation rather than the coefficients 
were actually the estimand of interest.
\citet{Graeb2011} use nested random intercepts to model the spatial 
(co)-variation due to \textit{communities}, \textit{provinces} and \textit{regions}
while including household and community characteristics as covariates associated 
with fixed effects. 
Their model M2 utilizes selected covariates on the household level, including
the household size (\textit{hh.size}), the kid-per-adult ratio (\textit{hh.kid.adult}), 
the youth-per-adult ratio (\textit{hh.youth.adult}), the age of the household head
(\textit{hh.age.head}), whether the household head was literate (\textit{hh.lit.head}), 
the literacy percentage of adult household members (\textit{hh.lit.adult}) and
whether the household head was christian (\textit{hh.christian}).
Selected community level covariates include the ethnic fragmentation in the community
(\textit{com.ethnic}), the adult literacy fraction in the community 
(\textit{com.lit.adult}), the community youth-adult ratio (\textit{com.youth.adult}), 
whether at least one household in the community has electricity (\textit{com.electricity})
and whether the community was urban (\textit{com.urban}). 
The authors utilize the ICC to judge the contribution of each spatial level to the 
variance in the outcome. 
However, the ICC ignores the variation induced by the covariates associated with fixed 
effects, additionally to ignoring the variance of the data-generating process of the 
covariates associated with random effects. 
The variation attributable to the mean structure is determined by the proportional change
of the ICC in a pure random intercept model compared to the ICC in a LMM with the
fixed effects.
Clearly, this is a rather heuristic approach, and a simultaneous modeling of the 
variance due to covariates associated with fixed and random effects is not 
possible this way.  

We fit model (\ref{LMM}) to represent M2 and illustrate how our
novel variance decomposition provides a formal approach to simultaneously judge the 
contribution of each spatial level and each household and community level covariate
to the variability of the logarithm of \textit{household expenditures}. 

We additively decompose 
$\syh = 0.883$ of the logarithm of \textit{household expenditures}. 
The estimate $\Rbd$ equals 
$0.375$ ($\approx 42\%$).
The estimate of the unbiased predictor for the explained variation of \textit{community},
\textit{province}, and \textit{region} associated with random effects is $\Rud=0.050$ 
($\approx 6\%$). 
The \textit{data set specific} explained variance $\Rud{}^{d}$ is $0.005$ 
($\approx 0.4\%$) and the
\textit{population} explained variance is $\Ru{}^{p}$ is $0.05$ ($\approx 5.6\%$). 
The estimate of the \textit{data set specific} covariance term $\Rbud$ represents
a substantial contribution of $0.058$ ($\approx 6.6\%$ which should not be ignored. 
The REML estimate $\seh=0.395$ ($\approx 45\%$) for the residual 
variance. 
This yields an adjusted coefficient of determination $\Rlmm\approx 55\%$, see
definition (\ref{R2lmm}).
The population coefficient $\Rlmm_{p}$ in equation (\ref{R2lmmp}) 
amounts to $\approx 52\%$ which approximately equals  $R_c^2$. 

In Table \ref{tab:income} we present relative contributions of the individual covariates, 
including \textit{data set specific} contributions, to the explained variance.
The combined explained variation by the spatial levels is about $9.5\%$ including
its share of $\Rbud$, see Equation (\ref{Rzexplicit}). 
The covariates associated with fixed effects that are most relevant with respect to
dispersion are household size ($8.2\%$) and variables associated with literacy
($22.6\%$ combined with bootstrap confidence interval $(19.6\%; 25.8\%)$ ). 
Characteristics on the household level explain a combined amount of $\approx 25\%$
and the characteristics on the community level explain about $24\%$ including
\textit{community} and its random effect.

\bigskip

\begin{table}[ht]
	\singlespace
	\captionsetup{font=small}
	\small
	\caption{\doublespacing
		Inequality in Burkina Faso.
		$\dot{R}^2_{(\cdot)}$ denotes parameter-wise contributions to dispersion relevance, 
		as defined in equation and (\ref{Rxexplicit}) and (\ref{Rzexplicit}), 
		relative to $\syh$.
		We include $95\%$ parametric bootstrap confidence intervals. 
		\label{tab:income}}
	{\tabcolsep=4.25pt
		\begin{tabular}{@{}llrrrrrr@{}}
				Effect Type & Covariable & Estimate & Std. Error & Pr($>$ $|$t$|$) & 
				$\dot{R}^2_{(\cdot)}$ (in \%) & 2.5\% & 97.5\%\\
			\hline
			Fixed & (Intercept) & 11.13 & 0.05 & $<$ 0.001 &       &  &  \\ 
			& \textit{hh.size} & -0.04 & 0.001 & $<$ 0.001 &   8.24 &  7.18 &  9.35 \\
			& \textit{hh.kid.adult} & -0.05 & 0.006 & $<$ 0.001 &   0.86 &  0.58 &  1.21 \\
			& \textit{hh.youth.adult} & -0.05 & 0.006 & $<$ 0.001 &  0.66 &  0.42 &  0.99 \\
			& \textit{hh.age.head} & -0.004 & $<$ 0.001 & $<$ 0.001 &   2.06 &  1.49 &  2.69 \\
			& \textit{hh.lit.head} & 0.35 & 0.023 & $<$ 0.001 &  7.96 &  6.88 &  9.30 \\
			& \textit{hh.lit.adult} & 0.20 & 0.023 & $<$ 0.001 &   4.96 &  3.86 &  6.13 \\
			& \textit{hh.christian} & 0.03 & 0.019 & 0.071 &   0.39 & -0.02 &  0.85 \\
			& \textit{com.ethnic} & 0.35 & 0.059 & $<$ 0.001 &   3.39 &  2.21 &  5.09 \\
			& \textit{com.lit.adult} & 0.50 & 0.070 & $<$ 0.001 &   9.68 &  6.96 & 12.67 \\
			& \textit{com.youth.adult} & -0.09 & 0.041 &  0.035 &   0.42 &  0.03 &  0.92 \\
			& \textit{com.electricity} & 0.14 & 0.041 & $<$ 0.001 &   3.05 &  1.30 &  4.95 \\
			& \textit{com.urban} & -0.15 & 0.046 & 0.001 &   4.05 &  1.25 &  6.21 \\
			Fixed total &    &  &  &  & 45.74 & 43.40 & 48.18 \\ 
			Random & \textit{community} & $\dot{\sigma}_{u_1}^2/\syh$ = 0.035 & - & - & 3.72 & 
			2.88 & 4.81 \\ 
			& \textit{province} & $\dot{\sigma}_{u_2}^2/\syh$ = 0.006 & - & - &   1.49 & -0.44 &  
			2.51 \\ 
			& \textit{region} & $\dot{\sigma}_{u_3}^2/\syh$ = 0.017 & - & - &   4.31 & -1.22 &  6.12 \\ 
			Random total &   &  &  &  &  9.52 &  2.50 & 11.00 \\ 
			Residual &  & $\dot{\sigma}_{\varepsilon}^2/\syh$ = 0.447 & - & - & 44.74 & 43.44 & 
			52.66 \\ 
	\end{tabular}}
\end{table}

\begin{figure}[ht]
	\singlespace    \centering
	\captionsetup{font=small}
	\small
	\includegraphics[width = 1.0\textwidth]{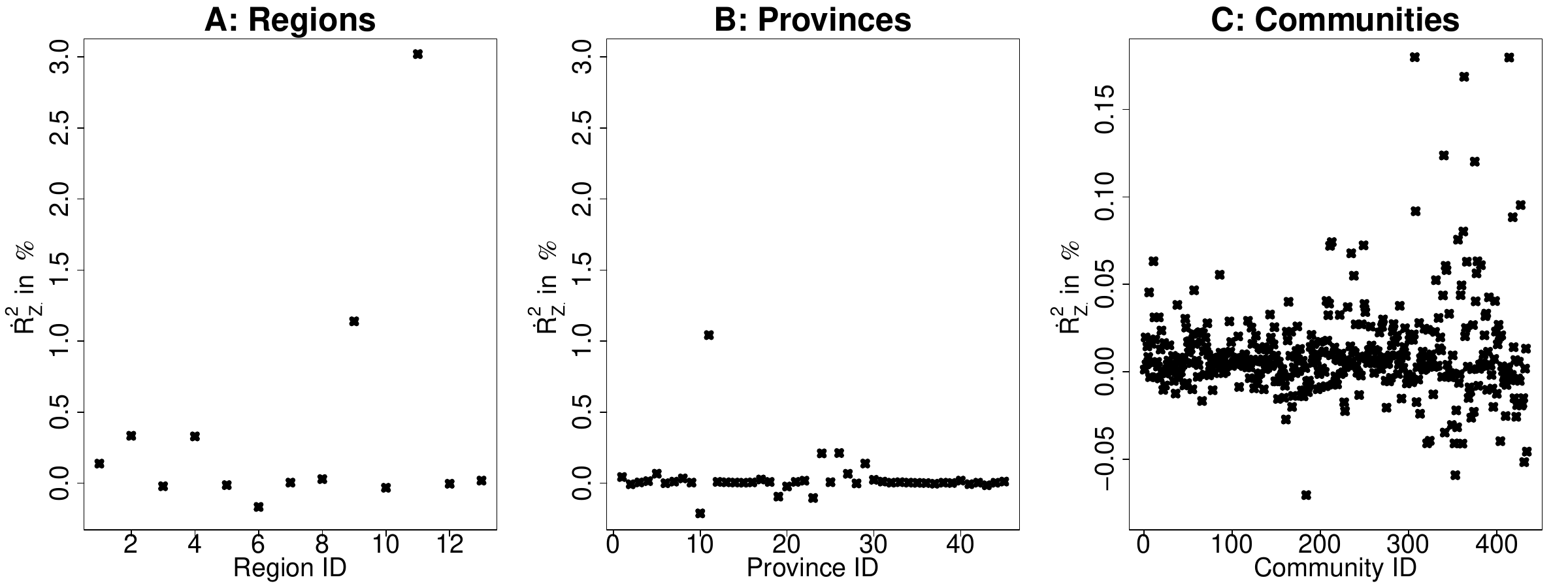}
	\caption{\doublespacing Parameter-wise Explained Variations for Burkina Faso Spatial Levels. 
		A: Region 11 (Centre/Ouagadougou) stands out. 
		B: Province 11 (Centre/Ouagadougou) stands out 
		C: Several communities stand out. Community IDs could not be assigned.  
		Note: The $y$-scale in C is smaller than in A and B. 
		\label{fig: bf}}
\end{figure}

In Figure \ref{fig: bf}, we illustrate the parameter-wise relevance profiles of the 
spatial levels \textit{region}, \textit{province} and \textit{community}. 
The \textit{region} as well as the \textit{province} with ID $11$ contribute 
prominently to the explained variance compared to the other \textit{regions} and
\textit{provinces}.
Interestingly, we identified ID 11 as Centre, which is the region in 
Burkina Faso with the highest population density and with the only large town in the 
country (Ouagadougou, the capital city). 
Unfortunately, we were not able to assign the \textit{community} ID's.

Our novel variance decomposition provides a formal tool to simultaneously
judge the relevance of hierarchical levels and covariates associated with fixed 
effects on the same scale. 
Thus, we were able to contrast the sum of contributions of multiple fixed effects 
covariates to the contribution of random effects capturing spatial structure. 
Income inequality, which is measured by the variability in household expenditures, 
could be associated with literacy as a main contributor while the spacial levels 
contributed less. 
Further, relevance profiles for \textit{province} and \textit{region} revealed that
the region ``Centre/Ouagadougou'' contributed the most to the spatial inequality 
while accounting for covariates associated with fixed effects.

\subsection{High-Dimensional GWAS - Arabidopsis Thaliana}
\label{subsec: gwas}

GWAS focus on finding associations between complex phenotypes 
and marker loci often spread on the whole genome.
In this context, model (\ref{LMM}) is widely applied to fit all markers 
on the genome simultaneously \citep[e.\,g.\,]{Zhou2012}. 
This enables a more realistic analysis than marker-wise comparisons because
a combined model allows for controlling confounding effects of the often
highly correlated markers.
It is of central interest to identify genetic variants that 
explain variation in a phenotype \citep{McCarthy2008}, as well as to attribute 
the observed variation to specific covariates, in particular to certain regions
of the genome such as chromosomes or even single
markers \citep[e.\,g.\,]{Pimentel2011}. 
This reinforces the need for a dispersion relevance profile in the LMM, 
in particular for genomic markers, which are represented by covariates 
associated with random effects.

The publicly available data set for the model organism  
``Arabidopsis thaliana'' \citep{Consortium2016} 
comprises $n=1057$ lines which were genotyped for $p=193,697$ 
single nucleotide polymorphism (SNP) markers split on $m=5$ chromosomes
($p_1=47,518$, $p_2=25,550$, $p_3=38,813$, $p_4=33,240$, $p_5=48,576$). 
We relate the phenotype flowering time at $10 {}^{\circ} C$, \textit{ft10}, 
to the marker genotypes on the $m=5$ chromosomes of the organism to investigate
the proportion of variation of \textit{ft10} each chromosome, and subsequently 
each individual SNP, explains.
To this end, we fitted a LMM (using R package \texttt{sommer} \citep{Pazar2017}) 
in which we model each chromosome as a
covariate associated with a random effect vector (of size $p_i$, $i = 1, \ldots, 5$). 
We include a fixed intercept, but no other covariates associated with fixed effects. 

The REML estimate of the unexplained variation
and the estimate of $\Rud$ contributed $7.75\%$ and $92.25\%$ to the 
sample variance of \textit{ft10}, respectively. 
The \textit{population} explained variation amounted to $47.80\%$ together with large
\textit{data set} specific contributions ($44.46\%$ in sum). 
The latter has already been shown to be the contribution of linkage disequilibrium, 
i.\,e.\,when covariates (markers) are correlated \citep{Schreck2019}. 

In Table \ref{tab: arab}, we give a detailed overview of the explained variation
between and among the five chromosomes, as derived in Equation (\ref{Rzexplicit}).
Chromosome five explains by far the largest amount of variance of \textit{ft10}
compared to the other chromosomes. 
In particular, the first three chromosomes explain very little variance by 
themselves. 
About one third of the explained variation can be attributed to effects 
between pairs of chromosomes where pairs that contain chromosome five contribute
the largest. 
Based on this high-level overview, we investigate single chromosomes in 
greater detail, in particular chromosome five. 

\bigskip

\begin{table}[ht]
	\singlespace
	
	\captionsetup{font=small}
	\small
	\caption{\doublespacing \textit{Arabidopsis thaliana}: Contributions of 
		each chromosome to the explained variation in \% (relative to $\syh$); 
		population contribution in parenthesis. 
		Diagonal: Isolated contributions of chromosomes represented
		by the first two terms in Equation (\ref{Rzexplicit}). 
		Off-diagonals: Contributions due to empirical correlations between
		SNP information of the chromosomes
		given by individual components of the sum in Equation (\ref{Rzexplicit}). 
		\label{tab: arab} }%
	{\tabcolsep=4.25pt
		\begin{tabular}{@{}c|lllll|r@{}}
				Chromosomes & $1$ & $2$ & $3$ & $4$ & $5$ & $\dot{R}^2_{Z_{j}}$\\
			\hline
			$1$ & $4.32 \, (3.94)$ & $0.63$ & $0.19$ & $1.06$ & $2.06$ &  $8.25 \, (3.94)$\\
			$2$ & $0.63$ & $6.43 \, (5.20)$ & $0.25$ & $1.57$ & $2.58$ & $11.46 \, (5.20)$\\
			$3$ & $0.19$ & $0.25$ & $2.67 \, (2.62)$ & $0.50$ & $1.06$  &$4.67 \, (2.62)$\\
			$4$ & $1.05$ & $1.57$ & $0.50$ & $14.11 \, (11.76)$ & $5.12$ & $22.35 \, (11.76)$\\
			$5$ & $2.06$ & $2.58$ & $1.06$ & $5.12$ & $34.70 \, (22.28)$ & $45.52 \, (22.28)$\\
			\hline
			$\dot{R}^2_{Z_{j}}$ & $8.25 \, (3.94)$ & $11.46 \, (5.20)$ & $4.67 \, (2.62)$ &
			$22.35 \, (11.76)$ & $45.52 \, (22.28)$ & $92.25 \, (47.80)$\\
			\hline
	\end{tabular}}
\end{table}

In Figure \ref{fig: arab5}A, we illustrate the coefficient 
p-values based on univariable LM's 
where single markers on chromosome $5$ are regressed on \textit{ft10} (Manhattan plot). 
However, these univariable models suffer strongly from confounding, in particular
because markers close on the genome tend to be highly correlated. 
Additionally, significance does not necessarily imply any relevance. 
To mitigate these issues, we introduce a novel dispersion relevance plot 
in Figures \ref{fig: arab5}B to visualize the parameter-wise contributions of 
individual SNP's to the observed variation in \textit{ft10} as described 
in the end of Section \ref{sec: part2}. 

\begin{figure}[h!]
	\singlespace    \centering
	\captionsetup{font=small}
	\small
	\includegraphics[width = 1.0\textwidth]{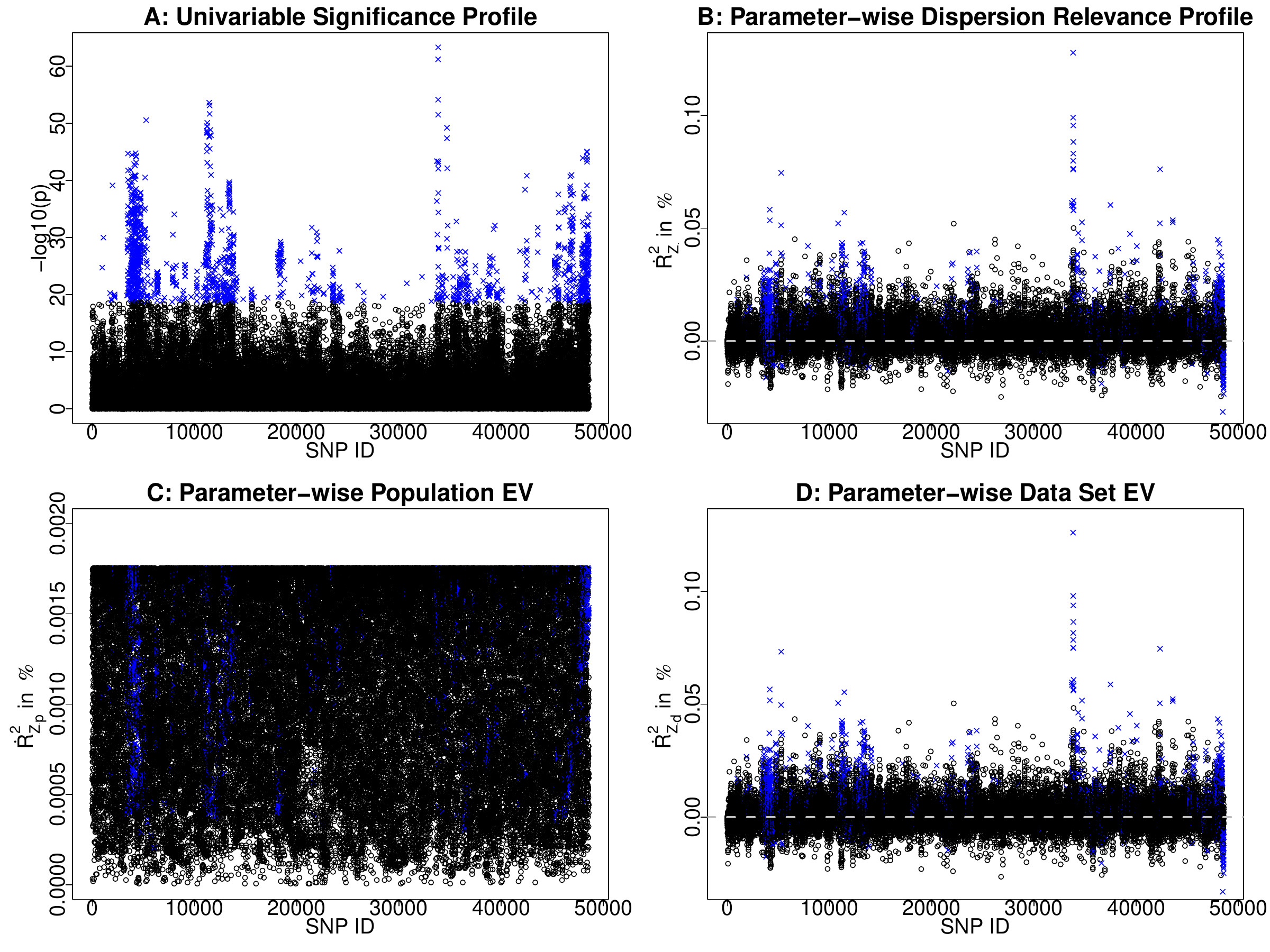}
	\caption{\doublespacing Arabidopsis Thaliana Chromosome $5$: 
		A: log-transformed p-values of the SNPs
		in single marker linear regression. 
		Areas with large values indicate regions in the genome with significant markers
		(threshold $10^{-14}$)
		(shape x in all four plots). 
		B: Parameter-wise relevance contributions, all markers have been fitted as 
		covariates associated with random effects in the LMM. 
		Areas with large values indicate regions in the genome with high relevance
		for the dependent variable \textit{ft10}. 
		C-D: Illustration of the \textit{population} explained variation (EV) and the 
		\textit{data set} specific
		explained variation relative to $\syh$. 
		The $y$-scale in C is remarkably smaller than in D. 
		\label{fig: arab5}}
\end{figure}

Several regions on chromosome $5$ indicated by the Manhattan-type plot in
Figure \ref{fig: arab5}A cannot be replicated to the same extent in the relevance 
profile in Figure \ref{fig: arab5}B, reflecting the above mentionend fact that
significance does not imply relevance. 
The region on the genome around SNP ID $34000$ is, however, prominent in both
the Manhattan-type plot as well as the relevance profile. 
Some significant SNPs in univariate analysis have negative contributions to 
the explained variation. 
This suggests that these regions might not be particularly important with 
respect to dispersion relevance in the fitted LMM. 
Naturally, in the context of this paper, we cannot give definitive answers of 
important regions or even individual SNPs. 
However, the simultaneous relevance assessment mitigates confounding and
problems concerned with significance-based analyses. 
Thus, it might be meaningful to direct the main focus on spikes present in 
both approaches and 
investigate differences. 

The population variation (Figure \ref{fig: arab5}C) is relatively constant 
and always positive over the whole chromosome, but comparably low for each marker. 
Differences arise solely to differences of the diagonal of $\Szh$, but are 
almost completely mitigated by the small value of $\suh$. 
These characteristics arise from the derivation of the population explained variance
as the unconditional expectation of $\sigma_r^2$ combined with the model assumptions
about the random effects, see Section \ref{sec: dsd} and Equation (\ref{LMM}). 
The data set specific explained variation (Figure \ref{fig: arab5}D) determines
the shape of the relevance profile which supports their importance in particular for 
parameter-wise or individual relevance assessments. 
The dispersion relevance profiles for the chromosomes $1-4$ are illustrated 
in the Supplementary Material, as well as an additional example (mice data) including 
coefficients associated with fixed effects.

\section{Discussion}
\label{sec4}

We have rigorously proven a novel proper decomposition of the observed sum of 
squares in the variance components form of the LMM, where only estimates that
result from the model fit using REML are involved. 
The components of this decomposition have been shown to represent unbiased
estimators of empirical explained variances, which constitute easily interpretable
estimands within the scope of our paper.  
We have introduced a bias-corrected estimator for the explained variance 
of covariates associated with fixed effects. 
Our estimator of the explained variation of covariates associated with 
random effects consists of what we call a \textit{population} and a purely additive 
\textit{data set specific} contribution. 
The \textit{population} part often constitutes the main contribution to the explained
variation, while the \textit{data set specific} part is particularly important for 
the parameter-wise dispersion relevance profile. 
The explained variance attributed to correlations between the covariates
associated with fixed and random effects has a non-zero 
\textit{data set specific} contribution but zero \textit{population} mean. 
On the one hand, we have utilized our decomposition to define an extension 
of the adjusted coefficient of determination from the LM to the LMM. 
On the other hand, we have introduced a parameter-wise dispersion relevance 
assessment for covariates associated with fixed and random effects on a common scale 
within a single LMM fit. 
To the best of our knowledge, a similar result has not been reported before. 
Its usefulness as a complement to the standard LMM output for hierarchical or
multilevel data and its application as a novel relevance screening plot for 
high-dimensional genomic data has been illustrated in several real data examples. 

Throughout the paper we emphasize that the data-generating processes of the 
covariates (associated with fixed or random effects) take a key part in 
the generation of explained variation. 
For instance, we have illustrated that the explained
variation is induced by the data-generating process $X$ of the covariates
associated with fixed effects. 
As per definition, the fixed effects have no variance and contribute only as
weighting factors.
The LMM, however, comprises several sources of variation. 
Additionally to the (co-)variances of the data-generating 
processes of the covariates associated with fixed and random effects,
the variation due to the random effects is important. 
Decomposition (\ref{vardecomp}) illustrates how these different sources of
variation interact and lead to a full decomposition of $\syh$.

Decomposition (\ref{vardecomp}) and with it explained variances, 
adjusted coefficient of determination and the parameter-wise relevance ranking 
are foremost applicable for model (\ref{LMM}) with continuous response and 
independent random effects. 
An extension to generalized linear mixed models and LMMs with correlated
random effects is part of future research.

The homoscedastic LM (\ref{statmod}) implies that the 
observations have identical variance, which in turn is decomposed into 
explained and residual variance.
This gives rise to the proportion of variance of the
observations explained by the covariates, which can be viewed as 
a population parameter.
In most heteroscedastic models, including the LMM, there is no 
unique variance of the dependent variable to decompose, with the exception 
of purely random intercept models. 
We focus on the decomposition of the quadratic 
form $\syh$, which reflects an empirical variance approach.
To this end, the realizations of the random effects are estimated
via empirical BLUP (or equivalently parametric Gaussian empirical Bayes), 
which is the prominent strength of the LMM over marginal models like the 
general linear model \citep{Xu2003}. 
Explicit prediction of random effects is advocated, for instance, 
by \citet{Robinson1991} and is very common, amongst others, in GWAS, 
and in cases were subject-specific merit is of interest, e.\,g.\,in cattle breeding
\citep{Morris1983, Robinson1991}. 
In the particular case that these estimates are of interest, 
we believe they should also be utilized in the calculation of explained variation
leading to the \textit{data set specific} contributions. 
These are purely additive components of the explained variance, i.\,e.\,they 
can always be considered complementary to the \textit{population} contribution.
To the best of our knowledge, none of the existing approaches take the data set 
specific contribution into account. 

Similar to ANOVA, and derived in the orthogonal case, it is a favourable property
that the proportions of variance attributable to each of the covariates
sum to $R^2$, although it is not a strict requirement for a meaningful measure 
of parameter-wise relevance \citep{Darlington1968}. 
In the LMM, alternative approaches such as semi-partial $R^2$'s are only available
for the coefficients associated with fixed effects, because dropping covariates
associated with random effects invalidates model comparisons. 
Additionally, semi-partial $R^2$'s suffer from drawbacks such as multiple model 
comparisons and lack of interpretation in terms of explained variation as
they typically do not sum to the full model $R^2$. 
Consequently, our parameter-wise relevance assessment is based on 
decomposition (\ref{vardecomp}) and equal (co-)variance sharing 
\citep{Hoffmann1960, Pratt1987}. 
Different weights for the (co-)variance sharing 
are possible and easily implemented, however, in the absence of convincing 
alternatives, we prefer the straightforward equal weighting in this already
complex context. 

Parameter-wise contributions to the adjusted $R^2$ can be negative.
Although potentially confusing, this happens for example when the variance 
of an estimate is large while the effect estimate itself is small. 
Additionally, individual contributions might become negative if the corresponding
variable decreases the variance of the observations by a strong negative correlation
with other variables \citep{Pratt1987}. 
In such cases, second-guessing the model at hand might be reasonable. 
Moreover, using a single measure to find the explained variation is always
a simplification of a very complex situation and should therefore not be
expected to provide unambiguous answers in all scenarios \citep{Pratt1987}. 

Our approach does not assess \textit{independent} variance contributions of
covariates associated with fixed and random effects, but contributions
given a specific model context. 
Therefore, our parameter-wise covariate relevance assessment is not 
intended for variable selection purposes. 
As argued in \citet{Cantoni2021}, variable selection in the LMM should be based
on information criteria, whereas $R^2$ provides information about the 
fitted model. 
Similarly, \citet{Edwards2008} state that the concept of explained variation
has limited use in model building. 

Our novel parameter-wise assignment of shares of explained variances
constitutes, to the best of our knowledge, the only both practically viable 
and theory-based option to simultaneously judge the relevance of covariates 
associated with fixed and random effects on the same scale, the observed variation, 
within a single LMM fit. 
In our opinion, the parameter-wise explained variations may be valuable 
in complementing standard model outputs as a
descriptive dispersion relevance measure of the covariates associated with
both fixed and random effects in fitted LMMs. 

\bibliography{BibliographySiM2}

\clearpage

\end{document}